\documentclass[11pt]{article}
\usepackage[cm]{fullpage}
\usepackage[english]{babel}
\usepackage{cmap}
\usepackage[T2A]{fontenc}
\usepackage[utf8]{inputenc}

\usepackage{mystyle}

\title{Normal forms for linear displacement context-free grammars}

\author{Alexey Sorokin}

\newcommand{\cort}[2]{(#1_{0}, \ldots, #1_{#2})}
\newcommand{\rk}[1]{\mathrm{rk}(#1)}
\newcommand{\Vvar}{\mathrm{Var}}

\newcommand{\Tm}{\mathrm{Tm}}
\newcommand{\GrTm}{\mathrm{GrTm}}

\theoremstyle{plain}
\newtheorem{theorem}{Theorem}
\newtheorem{lemma}{Lemma}
\newtheorem{statement}{Statement}
\theoremstyle{definition}
\newtheorem*{definition}{Definition}
\newtheorem*{example}{Example}

\begin{document}

\maketitle

\begin{abstract}
In this paper we prove several results on normal forms for linear displacement context-free grammars. The results themselves are rather simple and use well-known techniques, but they are extensively used in more complex constructions. Therefore this article mostly serves educational and referential purposes.
\end{abstract}

\section{Displacement context-free grammars}

Displacement context-free grammars (DCFGs) are a reformulation of well-nested multiple context-free grammars. In this draft we use tuple notation. Let $\Sigma$ be a finite alphabet, then $\Sigma^*$ denotes the set of all words with letters in $\Sigma$, $\epsilon$ being the empty string. When $\Sigma$ is fixed, $\Theta_k$ denotes the set of all tuples of the form $(u_0, \ldots, u_k), \: u_i \in \Sigma^*$ and $\Theta = \bigcup\limits_{k \in \Bbn} \Theta_k$. We call $k$ the rank of the tuple $u = (u_0, \ldots, u_k)$ and denote it by $\rk{u}$. The length $|u|$ of a tuple $|u|$ is the sum of lengths of all its components, we denote by $\Theta^{(l)}$ the set of all tuples of length $l$. The notation $\Theta^{(\leq l)}$ and $\Theta^{(\geq l)}$ are also understood in a natural way.

We use the displacement context-free languages notation for well-nested MCFLs. We consider tuples of strings instead of gapped strings. Let $\Sigma$ be a finite alphabet, then $\Sigma^*$ denotes the set of all words with letters in $\Sigma$, $\epsilon$ being the empty string. When $\Sigma$ is fixed, $\Theta_k$ denotes the set of all tuples of the form $(u_0, \ldots, u_k), \: u_i \in \Sigma^*$ and $\Theta = \bigcup\limits_{k \in \Bbn} \Theta_k$. We call $k$ the rank of the tuple $u = (u_0, \ldots, u_k)$ and denote it by $\rk{u}$. The length $|u|$ of a tuple $|u|$ is the sum of lengths of all its components, we denote by $\Theta^{(l)}$ the set of all tuples of length $l$ and also write $\Theta^{(\leq l)}$ for $\bigcup\limits_{j \leq l} \Theta^{(j)}$.

On the set of tuples we define the concatenation operation $\cdot \colon \Theta_i \times \Theta_j \to \Theta_{i+j}$ and the countable set of intercalation operations $\odot_l \colon \Theta_i \times \Theta_j \to \Theta_{i+j-1}$:

%\vspace*{-8pt}
$$
\begin{array}{rcl}
\cort{x}{i} \cdot \cort{y}{j} & = & (x_0, \ldots, x_iy_0, \ldots, y_j) \\
\cort{x}{i} \odot_l \cort{y}{j} & = & (x_0, \ldots x_{l-1}y_0, y_1, \ldots, y_jx_l, \ldots, x_i)
\end{array}
$$

Let $N$ be a finite ranked set of nonterminals and $\mathrm{rk} \colon N \to \Bbn$ be the rank function. Let $Op_k = \{ \cdot, \odot_1, \ldots, \odot_k\}$, the set $Tm_k(N, \Sigma)$ of $k$-correct terms is defined as follows:
\begin{enumerate}
    \item $\forall j \leq k \: (\Theta_j \subset \Tm_k(N, \Sigma)$.
    \item If $\alpha, \beta \in \Tm_k$ and $\rk{\alpha} + \rk{\beta} \leq k$, then $(\alpha \cdot \beta) \in \Tm_k, \: \rk{\alpha \cdot \beta} = \rk{\alpha} + \rk{\beta}$.
    \item If $j \leq k, \: \alpha, \beta \in \Tm_k, \: \rk{\alpha} + \rk{\beta} \leq k + 1, \: \rk{\alpha} \geq j$, then \\ $(\alpha \odot_j \beta) \in \Tm_k, \: \rk{\alpha \cdot \beta} = \rk{\alpha} + \rk{\beta} - 1$.
\end{enumerate}

We assume that all the operation symbols are leftassociative and concatenation has greater priority then intercalation. We may also omit the $\cdot$ symbol, so the notation $A \odot_2 BC \odot_1 D$ means $(A \odot_2 ((B \cdot C)) \odot_1 D)$.

Let $\Vvar = \inbr{x_1, x_2, \ldots}$ be a countable ranked set of variables, such that for every $k$ there is an infinite number of variables having rank $k$. A context $C[x]$ is a term where a variable $x$ occurs in a leaf position, the rank of $x$ must respect the constraints of term construction. Provided $\beta \in Tm_k$ and $\rk{x} = \rk{\beta}$, $C[\beta]$ denotes the result of substituting $\beta$ for $x$ in $C$. A valuation function $\nu$ assigns words of rank $l$ to the variables of rank $l$ for any $l \leq k$ in an arbitrary way. It also maps all the elements of $\Theta$ to themselves. Interpreting the connectives from $Op_k$ as corresponding binary operations, we are able to calculate the valuation of every ground term (i.~e. containing no nonterminal occurrences). It is easy to prove that $\rk{\alpha} = \rk{\nu(\alpha)}$ holds for every $\alpha$. The set of $k$-correct ground terms is denoted by $\GrTm_k(\Sigma)$.

\begin{definition}
A $k$-displacement context-free grammar ($k$-DCFG) is a quadruple $G = \inang{ N, \Sigma, P, S }$, where $\Sigma$ is a finite alphabet, $N$ is a finite ranked set of nonterminals and $\Sigma \cap N = \varnothing, S \in N$ is a start symbol such that $rk(S) = 0$ and $P$ is a set of rules of the form $A \to \alpha$. Here $A$ is a nonterminal, $\alpha$ is a term from $Tm_k(N, \Sigma)$, such that $rk(A) = rk(\alpha)$.
\end{definition}

\begin{definition}
The derivability relation $\vdash_G \in N \times Tm_k$ associated with the grammar $G$ is the smallest reflexive transitive relation such that the facts $(B \to \beta) \in P$ and $A \vdash C[B]$ imply that $A \vdash C[\beta]$ for any context $C$. Let $L_G(A) = \{ \nu(\alpha) \mid A \vdash_G \alpha, \: \alpha \in GrTm_k\}$ denote the set of word, which are derivable from a nonterminal $A$, then $L(G) = L_G(S)$.
\end{definition}

\begin{example}
A $k$-DCFG $G_k = \inang{\inbr{S, T}, \inbr{a_i, b_i \mid i \in [0; k]}, P,S}$, where the set $P$ is defined below, derives the language $L_k = \{a_0^m b_0^m \ldots a_k^m b_k^m \}$.
\vspace*{-8pt}
$$
\begin{array}{rcl}
S & \to & \underbrace{( \ldots (}_{(k-1)\text{ times}} \!\!\! T \odot_1 \epsilon)\ldots) \odot_1 \epsilon\\
T & \to & a_0 (T \odot_1 (b_0, a_1) \ldots \odot_k (b_{k-1}, a_k)) b_k\\
T & \to & (\!\underbrace{\epsilon, \ldots, \epsilon}_{(k+1)\text{ times}}\!)
\end{array}
$$
\end{example}

In what follows we assume that all the string tuples which occur in term leaves belong to $\Theta^{(\leq 1)}$. Obviously, this constraint does not restrict the generative power of DCFGs.

\begin{definition}
A term is called linear if it contains zero or one occurrences of nonterminals. A grammar is linear if right sides of all its rules are linear terms.
\end{definition}

In this paper we study normal forms for linear DCFGs. The following result for DCFGs in general was obtained in \cite{Sorokin2013LFCS}.

\begin{theorem}
Every $k$-DCFG is equivalent to some $k$-DCFG $G = \inang{N, \Sigma, P, S}$ which has the rules only of the following form:
\vspace*{-8pt}
\begin{enumerate}
\item $A \to B \cdot C,  \mbox{ where } A \in N - \{X\},\: B, C \in N - \{S\}$,
\item $A \to B \odot_j C,\mbox{ where } j \leq k,\: A \in N - \{X\} ,\: B, C \in N - \{S, X\}$,
\item $A \to a, \mbox{ where } a \in \Sigma $,
\item $X \to (\epsilon, \epsilon)$,
\item $S \to \epsilon$.
\end{enumerate}
\end{theorem}

\section{Normal forms for linear DCFGs}

A valuation may be extended to variables and nonterminals by assigning every variable an arbitrary word of appropriate rank. When the valuation is fixed, the value of a context is calculated just like the term value. Two contexts are equivalent if they have the same value under all valuations. Obviously, if we replace the right-hand term in a grammar rule by an equivalent term, the generated language does not change. Basic equivalencies are listed below:

\begin{statement}\label{term-eq-basic}
The following ground multicontexts are equivalent:
\begin{enumerate}
\item $(x_1 \cdot x_2) \cdot x_3 \sim x_1 \cdot (x_2 \cdot x_3)$,
\item $(x_1 \cdot x_2) \odot_j x_3 \sim (x_1 \odot_j x_3) \cdot x_2$ if $j \leq \rk{x_1}$,
\item $(x_1 \cdot x_2) \odot_j x_3 \sim x_1 \cdot (x_2 \odot_{j - \rk{x_1}} x_3)$ if $\rk{x_1} < j \leq \rk{x_1} + \rk{x_2}$,
\item $(x_1 \odot_l x_2) \odot_j x_3 \sim (x_1 \odot_j x_3) \odot_{l + \rk{x_3} - 1} x_2$ if $j < l$,
\item $(x_1 \odot_l x_2) \odot_j x_3 \sim x_1 \odot_l (x_2 \odot_{j - l + 1} x_3)$ if $l \leq j < l + \rk{x_2}$,
\item $(x_1 \odot_l x_2) \odot_j x_3 \sim (x_1 \odot_{j - \rk{x_2} + 1} x_3) \odot_l x_2$ if $j \geq l + \rk{x_2}$.
\item $(\epsilon, \epsilon) \odot_1 x_1 \sim x_1$,
\item $x_1 \odot_j (\epsilon, \epsilon) \sim x_1$ for any $j \leq \rk{x_1}$.
\end{enumerate}
\end{statement}

\begin{lemma}
Every linear $k$-DCFG is equivalent to some $k$-DCFG with the rules only of the form
\begin{itemize}
\item $A \to uB$ or $A \to Bu$, $|u| \leq 1, \: u \neq \epsilon$,
\item $A \to B \odot_j u$, $|u| \leq 1$,
\item $A \to u$, $|u| \leq 1$,
\end{itemize}
\end{lemma}
\begin{proof}
Through the proof we define a well-formed term by the following definition:
\begin{itemize}
\item A nonterminal or an element of $\Theta^{(\leq 1)}$ is well-formed,
\item If $\alpha$ is a well-formed term, then any $k$-correct term of the form $u \alpha$ or $\alpha u$, where $u \in \Theta^{(\leq 1)}$ and $u \neq \epsilon$, is well-formed,
\item If $\alpha$ is a well-formed term, then any $k$-correct term of the form $\alpha \odot_j u$, where $u \in \Theta^{(\leq 1)}$, is well-formed,
\end{itemize}

It is sufficient to prove that every linear term $\alpha$ is equivalent to some well-formed term. This is done by induction on term construction using the basic equivalencies and the fact that $(u_0, \ldots, u_l) \odot_j \alpha \sim (u_0, \ldots, u_{j-1}) \alpha (u_j, \ldots, u_l)$ for any term $\alpha$.
\end{proof}

\newcommand{\sepone}{\mathrm{1}}
In what follows we sometimes denote the tuple $(\epsilon, \epsilon)$ by $\sepone$.

\begin{lemma}\label{nf-second}
Every linear $k$-DCFG $G$ is equivalent to some $k$-DCFG with the rules only of the form
\begin{itemize}
\item $A \to uB$ or $A \to Bu$, $|u| \leq 1, \: u \neq \epsilon$,
\item $A \to B \odot_j u$, $|u| \leq 1$,
\item $A \to u$, $|u| = 1$,
\item $S \to \epsilon$.
\end{itemize}
\end{lemma}
\begin{proof}
The proof is analogous to $\epsilon$-rules elimination in standard DCFGs. We assume that $G$ already has the form introduced in the previous lemma. We want to create a new grammar with the set of rules $P'$ where every nonterminal $A \neq S$ of rank $l$ generates all the tuples except $\underbrace{(\epsilon, \ldots, \epsilon)}_{l+1 \text{ times}} = \sepone^l$. At first we determine for every nonterminal, whether it generates the word $\sepone^l$, such nonterminals are called $\epsilon$-generating.  If $A$ generates only this word, then it is called strictly $\epsilon$-generating.

We process every element of the old set of rules $P$ by the following algorithm.
\begin{enumerate}
\item If the rule has the form $A \to B_l * u, \: * \in Op_k$ and $B$ is not strictly $\epsilon$-generating, then this rule is added to $P'$.
\item If the rule has the form $A \to B_l u$, $|u| = 1$ and $B$ is $\epsilon$-generating, then we also add the rules, obtained by binarizing the rule $A \to (\epsilon, \epsilon)^p u$.
\item If the rule has the form $A \to u B_l$, $|u| = 1$ and $B$ is $\epsilon$-generating, then we also add the rule $A \to u (\epsilon, \epsilon)^p$.
\item If the rule has the form $A \to B_l \odot_j u$, $|u| = 1$ and $B$ is $\epsilon$-generating, then we also add the rule $A \to (\epsilon, \epsilon)^{j-1} u (\epsilon, \epsilon)^{l-j}$.
\item We include to $P'$ all the rules in $P$ of the form $A \to u$, $|u| = 1$.
\item We also include the rule $S \to \epsilon$, if $\epsilon \in L(G)$. 
\end{enumerate}

The correctness of the constructed grammar is proved by standard induction on word length.

\end{proof}

\begin{lemma}\label{nf-third}
Every linear $k$-DCFG $G$ is equivalent to some $k$-DCFG with the rules only of the form
\begin{itemize}
\item $A \to uB$ or $A \to Bu$, $|u| \leq 1, \: u \neq \epsilon$,
\item $A \to B \odot_j u$, $|u| = 1$,
\item $A \to u$, $|u| = 1$,
\item $S \to \epsilon$.
\end{itemize}
\end{lemma}
\begin{proof}
We assume that $G$ already satisfies Lemma \ref{nf-second}. At first we want to remove the rules of the form $A \to B \odot_j \epsilon$. To reach this goal we create for every nonterminal $B$ and every $j \in [1; \rk{B}]$ its $j$-th bridge $\widehat{B}^j$ with the following properties: if $B$ generates the word $(u_0, \ldots, u_{j-1}, u_j, \ldots, u_l)$, then $\widehat{B}^j$ generates the word $(u_0, \ldots, u_{j-1} u_j, \ldots, u_l)$ and vice versa. Then we create bridges for the newly introduced nonterminals and so on. Since the bridged nonterminal has lower rank then the initial one, this process will terminate. 

To satisfy the declared properties we extend the grammar with the following rules. The notation $\widehat{u}^j$ denotes the word obtained from $u = (u_0, \ldots, u_{j-1}, u_j, \ldots, u_l)$ by removing the $j$-th gap. The subscript here and to the end of the paper marks the rank of the nonterminal.
\begin{enumerate}
\item For every rule $A \to u B$ we add the rule $\widehat{A}^j \to \widehat{u}^j B$ in case $j \leq \rk{u}$ and the rule $\widehat{A}^j \to u \widehat{B}^{l-j}$ in case $l \geq \rk{u}$.
\item For every rule $A \to B_r u$ we add the rule $\widehat{A}^j \to B \widehat{u}^{j-r}$ in case $j > r$ and the rule $\widehat{A}^j \to \widehat{B}^{j} u$ in case $j \leq r$.
\item For every rule $A \to B \odot_l u$ we add the rule $\widehat{A}^j \to \widehat{B}^j \odot_{l-1} u$ in case $j  < m$, the rule $\widehat{A}^j \to B \odot_{l} \widehat{u}^{j-l+1}$ in case $l \leq j < l + \rk{u}$ and $\widehat{A}^j \to \widehat{B}^{j - \rk{u} + 1} \odot_{l} u$ in case $j \geq l + \rk{u}$.
\item For every rule $A \to u$ we add the rule $\widehat{A}^j \to \widehat{u}^j$.
\item If the grammar contained the rule $S \to \epsilon$, we preserve this rule.
\end{enumerate}

Afterwards we remove replace every rule of the form $A \to B \odot_j \epsilon$ with the rule $A \to \widehat{B}^j$. We also replace all the rules of the form $A \to B \odot_j (\epsilon, \epsilon)$ by the rule $A \to B$ and then eliminate unary rules by standard procedure.

It remains to remove the rules of the form $A \to B \odot_j \sepone^l$ for $l \geq 2$. It is done analogously to the previous step. On the set of tuples we define the $j,l$-split operation $\widebar{u}^{j,l}$, which inserts the tuple $\sepone^l$ into the $j$-th gap of $u$, this operation is naturally extended to languages. For every nonterminal $B$ we introduce its $j,l$-split $\widebar{B}^{j,l}$ (in case $\rk{B} + l \leq k+1$) which generates the $(j,l)$-split of $L(B)$. We repeat this procedure until all nonterminals of rank less than $K$ have splitted versions. It is done just in the same way we have introduced the bridge nonterminals.

Now we replace every rule of the form $A \to B \odot_j \sepone^l$ by the rule $A \to \widebar{B}^{j,l}$ and eliminate unary rules as earlier. The lemma is proved.
\end{proof}

Finally, we want to eliminate tuples of length $0$ at all. For every natural $p$ we introduce an unary operation ${}_{/p}$, which transforms a tuple of the form $u = va\sepone^p$ to the string $u_{/p} = va$ in case $a \in \Sigma$, otherwise this operation is undefined. Informally, it deletes $p$ rightmost $\epsilon$ components of the tuple provided the rightmost fragment of the obtained tuple will be nonempty. The operation ${}_{\backslash p}$ is defined symmetrically. Both the operations are naturally extended from individual tuples to languages.

\begin{theorem}
Every linear $k$-DCFG $G$ is equivalent to some $k$-DCFG with the rules only of the form
\begin{itemize}
\item $A \to uB$ or $A \to Bu$, $|u| = 1$,
\item $A \to B \odot_j u$, $|u| = 1$,
\item $A \to u$, $|u| = 1$,
\item $S \to \epsilon$.
\end{itemize}
\end{theorem}
\begin{proof}
We assume that initial grammar $G = \inang{N, \Sigma, P, S}$ already satisfies Lemma \ref{nf-third}. We set $N' = \{ A_{/p} \mid A \in N, \, p \leq \rk{A} \}$, $S' = S_{/0}$ and construct the set $P'$ by the following procedure:
\begin{enumerate}
\item For every rule of the form $A \to u B$ we add the rule $A_{/p} \to u B_{/p}$ for all possible $p$.
\item For every rule of the form $A \to B (\sepone^q a \sepone^p)$ (every rule of the form $A \to Bu$ with $|u| = 1$ can be expressed so) we add the rule $A_{/p} \to B (\sepone^q a)$.
\item For every rule of the form $A \to B 1^q$ and every $p \geq q$, we add the rule $A_{/p} \to B_{/(p-q)}$.
\item For every rule of the form $A \to B \odot_j u$ and every  $p < \rk{B} - j$ we add the rule $A_{/p} \to B_{/p} \odot_j u$.
\item For every rule of the form $A \to B \odot_j (1^q a 1^r)$ we add the rule $A_{/p} \to B \odot_j (1^q a)$ with 
$p = r + (\rk{B} - j)$.
\item For every rule of the form $A \to a$ we add the rule $A_{/0} \to a$.
\item If $(S \to \epsilon) \in P$, then we also add the rule $S_{/0} \to \epsilon$.
\end{enumerate}

It is straightforward to check that $L(A_{/p}) = (L(A))_{/p}$, hence $L_(S_{/0}) = (L(S))_{/0} = L(S)$ as required.
We have eliminated rules of the form $A \to B \sepone^p$, the rules of the form  $A \to \sepone^p B$ are removed analogously. The theorem is proved.
\end{proof}

\bibliographystyle{plain}
\bibliography{../../../Science/Bibliography/sorokin,../../../Science/Bibliography/grammars}

\end{document}